\documentclass[final,12pt]{article}
\usepackage{hyperref}
\usepackage{upgreek}
\usepackage{textgreek}
\usepackage{amssymb,amsmath,amsthm}
\usepackage{cite}
\usepackage{tikz}
\usepackage[utf8]{inputenc}
\usetikzlibrary{arrows,backgrounds,decorations.pathmorphing,decorations.pathreplacing,positioning,fit}
\usepackage{enumerate}
\usepackage{multicol}
\usepackage{accents}
\usepackage{stackengine}
\usepackage[conditional,light,first,bottomafter]{draftcopy}
\draftcopyName{DRAFT\space\today}{130}
\draftcopySetScale{65}
\usepackage{lmodern}
\usepackage{microtype}
\usepackage{graphicx}
\usepackage{subcaption}
\usepackage{dsfont}
\usepackage[letterpaper,hmargin=2.5cm,vmargin=2.8cm]{geometry}
\usepackage{setspace}
\usepackage{colonequals}
\makeatletter
\renewcommand{\section}{\@startsection%
{section}%
{1}%
{0em}%
{1.7em}%
{1.2em}%
{\normalfont\large\centering\bfseries}}
\renewcommand{\@seccntformat}[1]%
{\csname the#1\endcsname.\hspace{0.5em}}
\makeatother

\numberwithin{equation}{section}
\newtheorem{theorem}{Theorem}[section]
\newtheorem{proposition}[theorem]{Proposition}
\newtheorem{lemma}[theorem]{Lemma}
\newtheorem{corollary}[theorem]{Corollary}
\theoremstyle{definition}
\newtheorem{definition}[theorem]{Definition}
\newtheorem{remark}[theorem]{Remark}

\newcommand{\abs}[1]{\mid{#1}\mid}
\newcommand{\cc}[1]{\overline{#1}}
\newcommand{\ceq}{\colonequals}
\newcommand\lrb[1]{\left\lbrace#1\right\rbrace}
\newcommand\lrp[1]{\left (#1\right)}
\newcommand\ip[2]{\langle {#1},{#2} \rangle}
\newcommand\no[1]{\| {#1} \|}
\newcommand\dis[1]{{\stackengine{-1pt}{ $\sim$ }{$\scriptscriptstyle {#1}$}{U}{c}{F}{T}{S}}}
\newcommand\cA[1]{\mathcal {#1}}
\DeclareMathOperator{\isosc}{isosc}

\makeatletter
\newcommand{\mathleft}{\@fleqntrue\@mathmargin0pt}
\newcommand{\mathcenter}{\@fleqnfalse}
\makeatother
\begin{document}
\begin{titlepage}
\title{The $k$-adjacency operators and adjacency Jacobi matrix on distance-regular graphs
\footnotetext{%
Mathematics Subject Classification(2010):
05C63, 
34L05, 
47B36. 
}
\footnotetext{%
Keywords: 
Spectral graph theory;
Distance-regular graphs;
Jacobi operators.
}
}
\author{
\textbf{Josu\'e I. Rios-Cangas}
\\
\small Departamento de Matem\'aticas\\[-1.6mm] 
\small Universidad Aut\'onoma Metropolitana, Iztapalapa Campus\\[-1.6mm]
\small San Rafael Atlixco 186, C.P. 09340, Iztapalapa, Mexico City.\\[-1.6mm] 
\small \texttt{jottsmok@xanum.uam.mx}
}
\date{}
\maketitle
\vspace{-4mm}
\begin{center}
\begin{minipage}{5in}
  \centerline{{\bf Abstract}} \bigskip
We deal in this work with a class of graphs, namely, the class of distance-regular graphs, in which  on the basis of $k$-adjacency operators, the adjacency operator $A$ of a distance-regular graph is identified as a Jacobi matrix. To get so, the set of the $k$-adjacency operators is recognized as a canonical basis in a certain Hilbert space, where the spectrum of the Jacobi matrix coincides with the support of the measure of $A$. The obtained identification permits a deeper spectral analysis of the graph. The finite-dimensional case is addressed by means of the extension theory of nondensely defined, symmetric linear operators.   
\end{minipage}
\end{center}
\thispagestyle{empty}
\end{titlepage}

\section{Introduction}\label{sec1}

\noindent The graph theory is closely related to the spectral theory of linear operators in Hilbert spaces, by recognizing the set of vertices of a graph $G$, as the canonical basis of a Hilbert space $\mathcal{H}$.  The adjacency matrix of $G$ acts as a linear operator in $\mathcal{H}$ and its spectrum allows analyzing the behavior of the graph. It is interesting to mention that it was not until the early eighties that the spectral theory of finite graphs was extended to the infinite case \cite{MR683222,MR657116} and nowadays, there are numerous topics with applications related with spectral theory of infinite graphs (see for example \cite{MR3869419,MR3891807,MR3617953}). Thereby, and as a motivation for the complex network theory \cite{MR3902704,MR3823031}, we tackle the spectral analysis of infinite graphs. We also handle the finite-dimensional case throughout extension theory of nondensely defined, symmetric linear operators; this is due to the fact that the selfadjoint extensions that we study here and the adjacency operator, have the same spectral distribution (q.v. Remark~\ref{rm:jinfnity-perturbation}). This viewpoint sheds some new light on the extension theory of finite graphs, and it is related to the solution of the classical truncated moment problem \cite[Sec.\, 10]{MR1318517}. It is worth noting that the densely defined condition of a linear operator in a Hilbert space can be relaxed, even when the Hilbert space is finite-dimensional, by using the theory of linear relations  \cite{MR0123188,MR0361889} (or multivalued linear operators \cite{MR1631548}). 

The concept of distance-regular graphs was introduced by N. Biggs in his seminal work \cite{MR1271140}, by realizing that these graphs held combinatorial symmetries and linear algebraic properties.  For an amenable and solid analysis in spectral theory, we address the notion of a distance-regular graph with respect to $k$-adjacency operators in a Hilbert space (q.v. Definition~\ref{def:regular-graph}). Basically,  the $k$-adjacency operator of a graph maps every vertex $v$ into the sum of vertices which are at distance $k$ from $v$. Particularly, a $k$-adjacency operator turns out an adjacency operator of another graph in the same Hilbert space. Besides, on distance-regular graphs, the $k$-adjacency operators obey a recurrence relation (see Theorem~\ref{th:recursive-Ak}), which allows these operators to be a basis in a certain Hilbert space and to get the identification of the adjacency operator with a Jacobi matrix. The advantage of using this identification and the theory of Jacobi operators  \cite{MR1711536} lies in the fact that we will develop an exhaustive spectral analysis of a distance-regular graph. We emphasize that this identification has been addressed in several works (see for example \cite{MR1271140,MR3617953}). This paper contains relatively new results and our viewpoint throws some new light on the theories of distance-regular graphs and $k$-adjacency operators, which are the basis of this article.

Let us summarize this paper as follows. We briefly discuss in Section~\ref{sec:relations} some standard facts on graphs, and we restrict our attention to bipartite graphs. Besides, we look more closely at the $k$-adjacency operators  and we lay out some practical concepts and results related to these operators.  Also, we present in Theorem~\ref{th:characterization-bigraph01} a characterization of bipartite graphs. 
We see in Section~\ref{sec:Adjacency-Jacobi} the regularity of every $k$-adjacency operator, and we introduce a notion of cyclicity (q.v. Definition~\ref{def:isoscycle-Koperator}).  Moreover, we give the notion of distance-regular graphs in terms of the $k$-adjacency operators. Theorem~\ref{th:recursive-Ak} shows that all the  $k$-adjacency operators are bounded, selfadjoint, regular, isoscyclical and obey a recurrence relation, which permits that the adjacency operator is identified as a Jacobi matrix in certain Hilbert space, in the sense that the support of the spectrum of the adjacency operator coincides with the spectrum of Jacobi operator (q.v. Theorem~\ref{th:Jacobi-Adjacency} and Remark~\ref{Rm:spectral-properties}). Section~\ref{sec:graphs-with-finite-diameter} is devoted to distance-regular graphs with finite diameter, and we address in this section the problem of finding the support of the spectrum of the adjacency operator, throughout extension theory for nondensely defined symmetric operators. Theorem~\ref{Thm:Jacobi-vs-A} allows determining the Jacobi operator that corresponds with the adjacency operator, and  Corollary~\ref{cor:multiplicity-eigenvalues-A} exhibits the so-called Biggs' formula, which provides the multiplicity of every eigenvalue of the adjacency operator.  Finally, we present in Section~\ref{sec:Examples} two standard examples to clarify the exposition of this work.

\section{Bigraphs and the $k$-adjacency operators}
\label{sec:relations}
\noindent
In this note any graph is assumed to be countable, undirected, unweighted, simple (without loops or multiple edges) and connected  (any pair of vertices is linked by edges), with a set of vertices
\begin{gather}
\label{eq:vertices}
 V\ceq\{\delta_i\}_{i\in{\mathbb N}}\,.
\end{gather}
Here ${\mathbb N}_0\ceq{\mathbb N}\cup\{0\}$ and $\partial (\delta_i,\delta_j )$ represents the distance between two vertices $\delta_i$, $\delta_j$, i.e., the minimal number of edges that join $\delta_i$ and $\delta_j$. 

\begin{definition} We say that two vertices $\delta_i,\delta_j$ are \emph{$k$-adjacent}, denoted by $\delta_i\dis k\delta_j$, $k\in {\mathbb N}_0$,  whenever $\partial (\delta_i,\delta_j )=k$. We simply say in the case $k=1$ that $\delta_i,\delta_j$ are adjacent (or neighbors) and we write $\delta_i\sim \delta_j$.
\end{definition}

The following subsets form a partition of $V$, with respect to $\delta_1$.
 \begin{gather}\label{eq:partition-bigraph}
\begin{split}
\lrb{\delta_i\in V\,:\, \partial (\delta_i,\delta_1)=2k}_{k\in{\mathbb N}_0}\,\,\,;\,\,\,
\lrb{\delta_i\in V\,:\, \partial (\delta_i,\delta_1)=2k+1}_{k\in{\mathbb N}_0}\,.
\end{split}
 \end{gather}

\begin{definition} A graph is called  \emph{bipartite} (or \emph{bigraph} for short) if  no two vertices belonging to the same subset of \eqref{eq:partition-bigraph} are adjacent (e.g., see Fig.~\ref{fig:star-grid}).
\end{definition}
\begin{figure}[h]\centering
    \includegraphics[width=15cm]{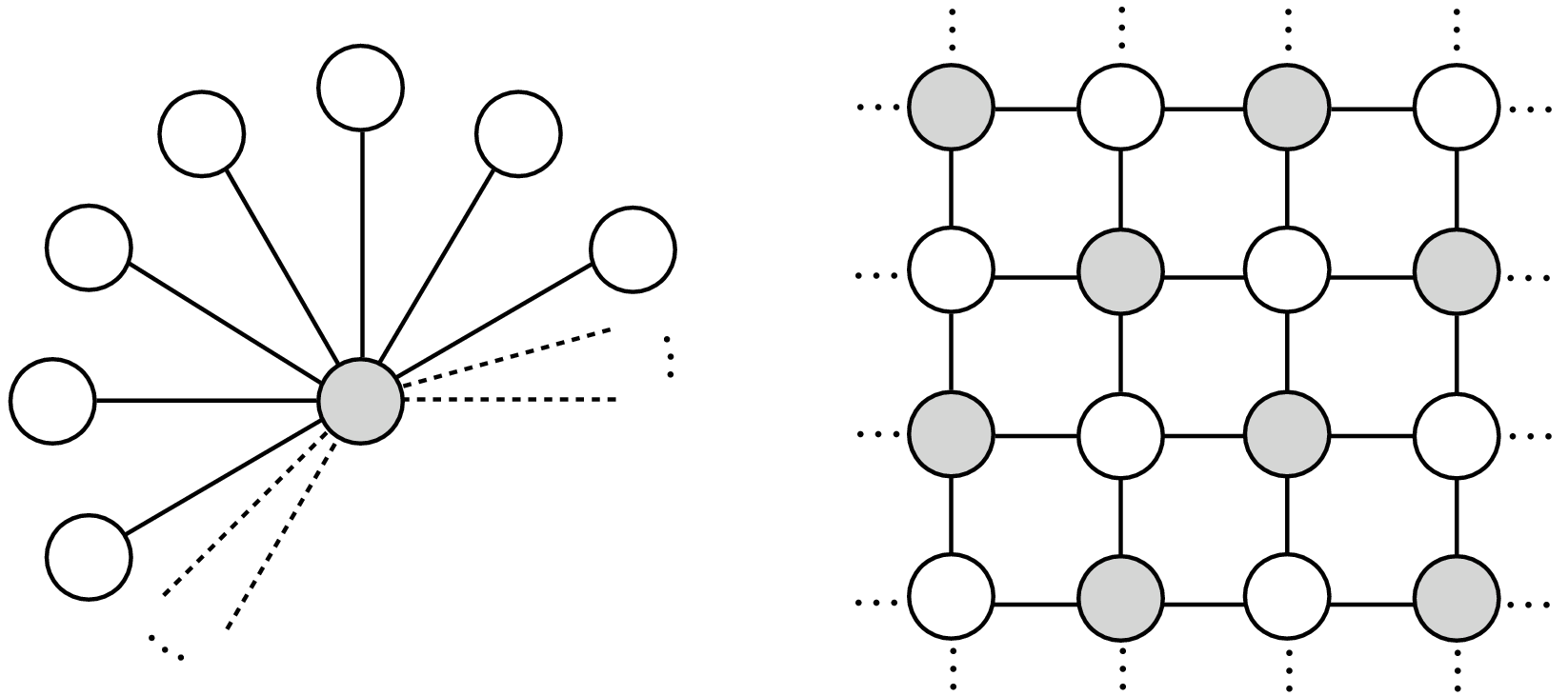}
    \caption{\centering Infinite star graph and grid graph are bipartite.}\label{fig:star-grid}
\end{figure}
\begin{proposition}\label{th-no-odd-cycles}
A graph is bipartite if and only if it has no cycles of odd length.
\end{proposition}
\begin{proof}
An odd cycle contains three vertices $\delta_i,\delta_j,\delta_l$, which satisfy $\delta_i\sim \delta_j$ and both are $n$-adjacent to $\delta_l$, for some $n\in{\mathbb N}$. This implies that both $\delta_i,\delta_j$ belong to the same subspace of \eqref{eq:partition-bigraph}. Therefore, a bigraph has no cycles of odd length.

On the other hand, if the graph is not bipartite, then one of \eqref{eq:partition-bigraph} contains two adjacent vertices $\delta_i,\delta_j$, both $k$-adjacent to $\delta_1$, for some $k\in{\mathbb N}$. Hence, any closed path which contains $\delta_1,\delta_i,\delta_j$, also contains an odd cycle. 
\end{proof}
From now on, any graph is also assumed  to be \emph{locally finite}, i.e., each of its vertices has a finite number of neighbors.  In Fig.~\ref{fig:star-grid}, the grid graph is locally finite, but the infinite star graph is not. 

For a given graph with set of vertices \eqref{eq:vertices},
we consider  the Hilbert space $\mathcal{H}\ceq l_2(V)$ of square-summable sequences with canonical basis $V$ and inner product $\ip{\cdot}{\cdot}$, being antilinear in the first argument.

For $k\in{\mathbb N}_0$, let $\tilde A_k$ be the linear operator acting on $V$ by
\begin{align}
\label{eq:adj-matrix}
\tilde A_k \delta_i=\displaystyle \sum_{\delta_j\dis k\delta_i}\delta_j\,,
\end{align}
which is well-defined, since we work on locally finite graphs. Some authors refer to \eqref{eq:adj-matrix} as the $k$-distance matrix \cite{MR1271140,MR3617953}. 

It is well-known that a linear operator $T$ in $\mathcal{H}$ is densely defined if its domain is dense in $\mathcal{H}$. Besides, $T$ is symmetric if 
\begin{gather*}
\ip{f}{Tf}\in{\mathbb R}\,,\quad\mbox{for all}\quad f\in\textrm{dom}\, T\,.
\end{gather*}
Moreover, it is selfadjoint if $T=T^*$, where $T^*$ is the adjoint of $T$.

\begin{proposition}\label{prop:k-operator}
Every $\tilde A_k$ is symmetric and densely defined in $\mathcal{H}$. 
\end{proposition}
\begin{proof} The proof is straightforward once we note that $\ip{\delta_i}{\tilde A_k\delta_j}=\ip{\tilde A_k\delta_i}{\delta_j}$, for any pair of vertices $\delta_i,\delta_j\in V$, and since $\textrm{span}\, V$ is dense in $\mathcal{H}$.
\end{proof}

\begin{definition} For $k\in{\mathbb N}_0$, the closure of $\tilde A_k$ is called  the \emph{$k$-adjacency} operator and denoted  by $A_k$. \end{definition}

Every $A_k$ is a densely defined, closed and symmetric linear operator. Also, $A_0$ is the identity operator $I$ and $A_1$ (which we only write $A$ for this operator) is the \emph{adjacency operator}. Besides, if the \emph{diameter} of the graph
\begin{gather}
\label{eq:diameter-graph}
d\colonequals\sup \lrb{\partial(\delta_i,\delta_j)\,:\,\delta_i,\delta_j\in V}<\infty\,,
\end{gather}
then $A_k=0$, for all $k>d$. 
\begin{remark}
We point out that the number  ${\no{A_k\delta_i}}^2$, with $k\in{\mathbb N}_0$, represents how many $k$-adjacent vertices $\delta_i$ has.
\end{remark}

\begin{definition}
A graph is called \emph{uniformly locally finite}, with bound $m<\infty$, if 
\begin{gather*}
{\no{A\delta_i}}^2\leq m\,, \quad\mbox{for all}\quad \delta_i\in V\,.
\end{gather*}
\end{definition}

\begin{remark}\label{re:bounded-condition}
The property to be uniformly locally finite characterizes the bounded condition of the adjacency operator.  Namely, $A$ is bounded if and only if its graph (not necessarily connected) is  uniformly locally finite with bound $m$. In this case, $\no A\leq m$  \cite[Th.\,3.2]{MR683222}.
\end{remark}

The below result uses the fact that if a closed densely defined operator in $\mathcal{H}$ is bounded, then it belongs to $\cA B(\mathcal{H})$ (the class of bounded operators defined on the whole space).
\begin{proposition}\label{prop:bounded-Ak-uniformly}
On uniformly locally finite graphs, $A_k$ belongs to $\cA B(\mathcal{H})$ and, hence, it is selfadjoint, for every $k\in{\mathbb N}_0$. 
\end{proposition}
\begin{proof}The case $k=0$ is simple, and from Remark~\ref{re:bounded-condition} $A$ is bounded. So, for $k\geq 2$ and $\delta_i$ fixed, one has  that $A^k\delta_i$ is the sum of vertices connected to $\delta_i$, by a walk of $k$-steps. In particular, the sum of vertices $k$-adjacent to $\delta_i$. Then, 
\begin{gather}\label{eq:condition-Ak}
{\no{A_k\delta_i}}^2\leq{\no{A^k\delta_i}}^2\leq{\no A}^{2k}\,, \quad\mbox{for all}\quad \delta_i\in V\,.\end{gather}
We conclude the proof using Remark \eqref{re:bounded-condition}, bearing in mind that $A_k$ is the adjacency operator of another  graph, which due to \eqref{eq:condition-Ak} is uniformly locally finite.
\end{proof}

It is a well-known fact that the spectrum $\sigma(T)$  of 
a  selfadjoint operator $T$ is a real subset and is the complement in ${\mathbb C}$ of the \emph{regular} set 
\begin{align*}
\rho(T)\ceq\lrb{\zeta\in{\mathbb C}\,:\,(T-\zeta I)^{-1}\in\cA B(\mathcal{H})}\,.
\end{align*}
Moreover, $\sigma(T)=\sigma_d(T)\cup\sigma_c(T)$, where 
\begin{align*}
\sigma_p(T)&\ceq\lrb{\zeta\in{\mathbb R}\,:\,\ker(T-\zeta I)\neq\{0\}}& \mbox{(\emph{point spectrum})}\\
\sigma_c(T)&\ceq\lrb{\zeta\in{\mathbb R}\,:\,\textrm{ran}\,(T-\zeta I)\neq\cc{\textrm{ran}\,(T-\zeta I)}}& \mbox{(\emph{continuous spectrum})}
\end{align*}

Let us decompose the Hilbert space into  $\mathcal{H}={\cA U}\oplus{\cA V}$, where  
 \begin{align*}
 \cA U&\ceq \cc{\textrm{span}\,\lrb{\delta_i\in V\,:\, \partial (\delta_i,\delta_1)=2k}}_{k\in{\mathbb N}_0}\,;\\
\cA  V&\ceq\cc{\textrm{span}\,\lrb{\delta_i\in V\,:\, \partial (\delta_i,\delta_1)=2k+1}}_{k\in{\mathbb N}_0}\,, 
 \end{align*}
which are the closure of the linear envelope of the sets given in \eqref{eq:partition-bigraph}.

\begin{theorem}\label{th:characterization-bigraph01}
On uniformly locally finite graphs, the following are equivalent:
\begin{enumerate}[\upshape(i)]
\item\label{eq:bip-ulf-1} The graph is bipartite.
\item\label{eq:bip-ulf-2} For every $f,g\in \cA U$ (or equivalently $f,g\in \cA V$),  
\begin{gather}\label{eq:general-condition-isoscycle}
\ip{AA_k f}{A_k g}=0\,,\quad \mbox{for all $k\in{\mathbb N}_0\,.$}
\end{gather}
\item\label{eq:bip-ulf-3}  The adjacency operator holds  
\begin{align}\label{eq:k-distance-bigraph}
A\cA U\subset\cA V\quad;\quad A\cA V\subset\cA U\,. 
\end{align}
\item\label{eq:bip-ulf-4} If $\zeta\in\sigma(A)$ then so does $-\zeta$, viz. 
$\sigma(A)$ is symmetric about zero. 
\end{enumerate}
\end{theorem}
\begin{proof}
\eqref{eq:bip-ulf-1}$\Rightarrow$\eqref{eq:bip-ulf-2} Since every $A_k$ is bounded, it is sufficient to prove \eqref{eq:general-condition-isoscycle} on $V$. For every  $\delta_i,\delta_j \in \cA U$, one has that $A_k\delta_i$ and $A_k\delta_j$ belong to the same subspace, either $\cA U$ or $\cA V$. Therefore, inasmuch as $\cA U$ and $\cA V$ do not contain adjacent vertices, one yields $\ip{AA_k\delta_i}{A_k\delta_j}=0$.
 
 \eqref{eq:bip-ulf-2}$\Rightarrow$\eqref{eq:bip-ulf-3} Let $\delta_i\in\cA U$. If $A\delta_i\notin \cA V$ then there exists $\delta_j\in \cA U$ such that $\delta_j\sim\delta_i$, which implies that $\delta_j,\delta_i$ are both $2k$-adjacent to $\delta_1$, for some $k\in{\mathbb N}$. Thus, 
\begin{align*}
\ip{AA_{2k}\delta_1}{A_{2k}\delta_1}=\displaystyle \sum_{\delta_t,\delta_s\dis {2k}\delta_1}\ip{A\delta_t}{\delta_s}\geq 2\,,
\end{align*}
a contradiction with  \eqref{eq:general-condition-isoscycle}, since $\delta_1\in\cA U$. Hence, $A\delta_i\in\cA V$ and  due to $A$ is bounded, $A\cA U\subset\cA V$.  The proof of $A\cA V\subset\cA U$ follows the same above lines. 

\eqref{eq:bip-ulf-3}$\Rightarrow$\eqref{eq:bip-ulf-4}  Since $A\in\cA B(\mathcal{H})$ and by \eqref{eq:k-distance-bigraph}, for any $f_1+f_2\in\ker (A-\zeta I)$, with $f_1\in\cA U$,  $f_2\in\cA V$, one has that $Af_1=\zeta f_2$ and  $Af_2=\zeta f_1$, which by a simple computation  $f_1-f_2\in\ker (A+\zeta I)$, viz.  $\ker (A-\zeta I)$ and $\ker (A+\zeta I)$ are in one-to-one correspondence. Using the last reasoning, we now proceed by contraposition. If $-\zeta\in\rho(A)$, i.e., $(A+\zeta I)^{-1}\in\cA B(\mathcal{H})$, then $(A-\zeta I)^{-1}$ is a linear operator, which is closed, in view of $A$ is closed. So, for every $h_1+h_2\in\mathcal{H}$, with $h_1\in\cA U$,  $h_2\in\cA V$, there exist $g_1\in\cA U$,  $g_2\in\cA V$, such that 
\begin{gather*}
A(g_1-g_2)+\zeta(g_1-g_2)=-h_1+h_2\,.
\end{gather*}
Then, by \eqref{eq:k-distance-bigraph}, one has $Ag_1-\zeta g_2=h_2$ and $Ag_2-\zeta g_1=h_1$, whereby $h_1+h_2$ belongs to $\textrm{dom}\, (A-\zeta I)^{-1}$. Hence, $(A-\zeta I)^{-1}\in\cA B(\mathcal{H})$, i.e.,  $\zeta\in\rho(A)$.

\eqref{eq:bip-ulf-4}$\Rightarrow$\eqref{eq:bip-ulf-1} Since $A$ is selfadjoint, we may consider its spectral measure $E_A$ and
\begin{gather*}
\mu_{A,\delta_i}(B)=\ip{\delta_i}{E_A(B)\delta_i}\,,\quad (\delta_i\in V)
\end{gather*}
which denotes a probability measure defined on the $\sigma$-algebra of Borel subsets of ${\mathbb R}$. Besides, this measure is symmetric since $\sigma(A)$ is symmetric. In this fashion, for each odd $m\in{\mathbb N}$,
\begin{gather*}
\ip{\delta_i}{A^m\delta_i}=\int x^md\mu_{A,\delta_i}=0\,,
\end{gather*}
wherefrom it follows that there are no closed paths of odd length, in particular, cycles of odd length. Hence, the graph is bipartite as a consequence of Proposition~\ref{th-no-odd-cycles}.
 \end{proof}

\begin{remark} On uniformly locally finite bigraphs, the property \eqref{eq:k-distance-bigraph} implies that  the adjacency operator is decomposed into $A= B\oplus{B^*}$, where \begin{gather*}
B=A_{\upharpoonright_{\cA U}}\colon\cA U\to\cA V\,.
\end{gather*}
\end{remark}

\section{Distance-regular graphs and the adjacency Jacobi operator}
\label{sec:Adjacency-Jacobi}
\noindent
Let us introduce some concepts before working on distance-regular graphs, which will be useful  in the sequel. 
\begin{definition}\label{def:regular-graph}
For $k\in{\mathbb N}_0$,  we say that $A_k$ is \emph{regular}, with degree $\deg A_k\in{\mathbb N}_0$, if 
\begin{gather*}
{\no {A_k\delta_i}}^2=\deg A_k\,,\quad\mbox{for all}\quad \delta_i\in V\,,
\end{gather*}
viz. all the vertices have the same number $\deg A_k$ of $k$-adjacent vertices.
\end{definition}

A vertex $\delta_i$ it said to have a \emph{$k$-isoscycle}, if there exist two adjacent vertices such that they are both $k$-adjacent to $\delta_i$. 
In such a case, $\delta_i$ belongs to an odd closed path of diameter equal $k$. Moreover, the number of $k$-isoscycles of $\delta_i$ is determined by 
$\ip{AA_k\delta_i}{A_k\delta_i}/2\in{\mathbb N}_0$.

\begin{definition}\label{def:isoscycle-Koperator}
For $k\in{\mathbb N}_0$, we call $A_k$ \emph{isoscyclical}, with isoscycle $\isosc A_k\in{\mathbb N}_0$, if 
\begin{gather*}
\frac12\ip{AA_k\delta_i}{A_k\delta_i}=\isosc A_k\,,\quad\mbox{for all}\quad \delta_i\in V\,,
\end{gather*}
viz. every vertex has the same number $\isosc A_k$ of $k$-isoscycles vertices. 
\end{definition}

It is worth pointing out that not all the operators $A_k$ are necessarily regular or isoscyclical, if one is.

\begin{definition}
A graph is called \emph{distance-regular} if there exists a sequence $\{(a_n,b_n)\}_{n\in{\mathbb N}}\subset{\mathbb N}^2$, such that for any pair of $k$-adjacent vertices $\delta_i,\delta_j$, with $k\in{\mathbb N}_0$, the following holds
\begin{align}\label{eq:cond-dist-regular}
\begin{split}
  \ip{A_{k-1}\delta_i}{A\delta_j}&=a_k\,,\\
\ip{A_{k+1}\delta_i}{A\delta_j}&=b_{k+1}\,.
  \end{split}
\end{align}
The sequence $\{(a_n,b_n)\}_{n\in{\mathbb N}}$ is known as  the \emph{intersection} of the graph.
\end{definition}

\begin{theorem}\label{th:recursive-Ak}
A distance-regular graph with intersection $\{(a_n,b_n)\}_{n\in{\mathbb N}}$, is uniformly locally finite and its $k$-adjacency operators are bounded and selfadjoint, for all $k\in{\mathbb N}_0$. Moreover, these operators hold the following difference equation:
\begin{gather}
\label{eq:recurrence-Ak}
AA_k=a_{k+1}A_{k+1}+\alpha_kA_k+b_kA_{k-1}\,,\quad \text{with }A_{-1}=0,
\end{gather}
where $\alpha_k=\deg A-(a_k+b_{k+1})$ and  $\alpha_0=0$. Furthermore, for $k>0$, every $A_k$ is regular and isoscyclical, with
\begin{gather}\label{eq:degree-Ak-distance}
\deg A_k=\displaystyle\prod_{n=1}^{k}\frac{b_n}{a_n}\,\quad\mbox{and}\quad  
\isosc A_k=\frac{\alpha_k}{2}\displaystyle\prod_{n=1}^{k}\frac{b_n}{a_n}\,.
\end{gather}
\end{theorem}
\begin{proof}
The first part of the statement is straightforward by Proposition~\ref{prop:bounded-Ak-uniformly}, once we note 
by virtue of \eqref{eq:cond-dist-regular} that $A$ is regular, with $\deg A=b_1$. To prove \eqref{eq:recurrence-Ak},  we regard two vertices $\delta_i\dis r\delta_j$, with $r\in{\mathbb N}_0$. If $\abs{r-k}>1$ then $\ip{AA_k\delta_i}{\delta_j}=\ip{A_k\delta_i}{A\delta_j}=0,$ which can be nonzero whenever $r\in\{k+1,k,k-1\}$. So, at a suitable $r$, taking into account \eqref{eq:cond-dist-regular},
\begin{gather}\label{eq:recurrence-delta}
\ip{A_k\delta_i}{A\delta_j}=a_{k+1}\ip{A_{k+1}\delta_i}{\delta_j}+\alpha_k\ip{A_k\delta_i}{\delta_j}+b_k\ip{A_{k-1}\delta_i}{\delta_j}\,,
\end{gather}
whence if $r=k$, then 
\begin{align*}
\alpha_k&=\ip{A_k\delta_i}{A\delta_j}\\&=\ip{A\delta_j}{A\delta_j}-\ip{A_{k-1}\delta_i}{A\delta_j}-\ip{A_{k+1}\delta_i}{A\delta_j}\\&=\deg A-(a_k+b_{k+1})\,.
\end{align*}
Hence,  \eqref{eq:recurrence-delta} implies \eqref{eq:recurrence-Ak}, since every $A_k$ is continuous. Now, for $k\in{\mathbb N}$ and $\delta_i\in V$, it follows by \eqref{eq:recurrence-Ak} that 
\begin{align*}
\ip{A_{k}\delta_i}{A_{k}\delta_i}&=\frac{1}{a_{k}}\ip{AA_{k-1}\delta_i}{A_{k}\delta_i}\\
&=\frac{1}{a_{k}}\ip{A_{k-1}\delta_i}{AA_{k}\delta_i}=\frac{b_{k}}{a_{k}}\ip{A_{k-1}\delta_i}{A_{k-1}\delta_i}\,,
\end{align*}
which recursively implies ${\no{A_k\delta_i}}^2=\prod_{n=1}^kb_n/a_n$. Also,  \eqref{eq:recurrence-Ak} produces 
\begin{align*}
\isosc A_k&=\frac12\ip{AA_{k}\delta_i}{A_{k}\delta_i}=\frac12\alpha_k\ip{A_k\delta_i}{A_k\delta_i}\,,\end{align*} whence one infers  \eqref{eq:degree-Ak-distance}.
\end{proof}

From now on, any graph is assumed to be distance-regular, which means that the $k$-adjacency operators are bounded, selfadjoint, regular and isoscyclical, for all $k\in{\mathbb N}_0$.

\begin{remark}\label{rem:Ak-polynomial}
Theorem~\ref{th:recursive-Ak} claims that every $A_k$ is a polynomial at $A$, of degree $k\in{\mathbb N}_0$. Indeed, $A_0=A^0$, $A_1=A$ and by \eqref{eq:recurrence-Ak},
\begin{gather}
\label{eq:polynomial-A_k-wrA}
A_{k}=\frac{1}{a_{k}}\Big(AA_{k-1}+\left(a_{k-1}+b_{k}-\deg A\right)A_{k-1}-b_{k-1}A_{k-2}\Big)\,,\quad k\geq2\,. 
\end{gather}
Besides, the intersection sequence is bounded. Actually,  \eqref{eq:degree-Ak-distance} implies $\alpha_k\geq0$ and consequently $\deg A\geq a_k+b_{k+1}$, for all $k\in{\mathbb N}$. Hence, $\{a_n\}_{n\in{\mathbb N}},\{b_n\}_{n\in{\mathbb N}}$ are bounded as well as $\{(a_n,b_n)\}_{n\in{\mathbb N}}$.
\end{remark}

\begin{proposition}\label{prop:polynomial-degree-A_k}
The $k$-adjacency operator  (seen as a polynomial at $A$) holds 
\begin{gather}\label{eq:polynomial-degree-A_k}
A_k(\deg A)=\deg A_k\,,\quad\mbox{for all}\quad k\in{\mathbb N}_0\,. 
\end{gather}
\end{proposition}
\begin{proof}
We will proceed by induction on $k$. Clearly, $A_0(\deg A)=\deg A_0$ and  $A_1(\deg A)=\deg A_1$. Then, we may suppose that \eqref{eq:polynomial-degree-A_k} holds for $k-1$. Note that \eqref{eq:degree-Ak-distance} implies $\deg A_{k}={b_{k}}\deg A_{k-1}/{a_{k}}$. In this fashion by \eqref{eq:polynomial-A_k-wrA}, 
\begin{align*}
A_{k}(\deg A)=&\,\frac{1}{a_{k}}\Big((\deg A) \deg  A_{k-1}+(a_{k-1}+b_{k}-\deg A)\deg A_{k-1}\\&-b_{k-1}\deg A_{k-2}\Big)=\frac{1}{a_{k}}(b_{k}\deg A_{k-1})\,,
\end{align*}
which yields \eqref{eq:polynomial-degree-A_k}.
\end{proof}

For a set of vertices $W\subset V$, let $\partial W$ denote the set of edges incident with exactly one vertex of $W$. 
\begin{definition}
The \emph{isoperimetric constant} of a graph is 
$\inf{\abs{\partial W}}/{\abs W}$, where the infimum is taken over all nonempty finite subsets of vertices. 
\end{definition}
The following assertion relies on the fact that, when the adjacency operator $A$ is regular, the isoperimetric constant is equal to zero if and only if the norm of $A$ satisfies $\no A=\deg A$ \cite[Th.\, 2.1 and Cor.\,3.3]{MR924236}.

\begin{corollary}\label{coro:bound-of-Ak}
The norm of every $k$-adjacency operator  holds
\begin{gather}\label{eq:norm-dist-regular}
\no {A_k}\leq \deg A_k\,,\quad k\in{\mathbb N}_0\,. \end{gather}
Moreover, the isoperimetric constant is zero if and only if the equality in \eqref{eq:norm-dist-regular} holds, for all $k\in{\mathbb N}_0$ .\end{corollary}
\begin{proof}
The first part readily follows from Remark~\eqref{re:bounded-condition}, inasmuch as, by Theorem~\ref{th:recursive-Ak}, every $A_k$ is regular and corresponds to an adjacency operator of another graph in the same space, which clearly  is uniformly locally finite with bound  $\deg A_k$.

Now, if the isoperimetric constant is equal to zero, then  $\deg A\in\sigma(A)$. So, the spectral mapping theorem  and Proposition~\ref{prop:polynomial-degree-A_k} claim that  $\deg A_k\in\sigma(A_k)$, which implies the equality in \eqref{eq:norm-dist-regular}. The converse is straightforward. 
\end{proof}

For simplicity of notation in the sequel, we write
\begin{gather}\label{eq:arthonormal-Ak}
{\mathbb A}_k\ceq\frac1{\sqrt{\deg A_k}}A_k\,,\quad (k\in{\mathbb N}_0)
\end{gather}
which is a polynomial at $A$ of degree $k$ (v.s. Remark~\ref{rem:Ak-polynomial}).
\begin{proposition}\label{prop:recursive-normalized}
If $\{(a_n,b_n)\}_{n\in{\mathbb N}}$ is the intersection sequence of the graph, then the following recursive equation holds:
\begin{gather}\label{eq:recursive-Jacobi}
A{\mathbb A}_k=\sqrt{a_{k+1}b_{k+1}}{\mathbb A}_{k+1}+\alpha_k{\mathbb A}_k+\sqrt{a_{k}b_{k}}{\mathbb A}_{k-1}\,,\quad k\in{\mathbb N}_0\,,
\end{gather}
where ${\mathbb A}_{-1}=0$ and $\alpha_k=\deg A-(a_k+b_{k+1})$, with  $\alpha_0=0$.
\end{proposition}

\begin{proof}
It follows  from \eqref{eq:recurrence-Ak} that  
\begin{align*}
A{\mathbb A}_k&=\frac1{\sqrt{\deg A_k}}\left(a_{k+1}A_{k+1}+\alpha_kA_k+b_kA_{k-1}\right)\\ 
&=a_{k+1}\sqrt{\frac{\deg A_{k+1}}{\deg A_{k}}}{\mathbb A}_{k+1}+\alpha_k{\mathbb A}_k+
b_{k}\sqrt{\frac{\deg A_{k-1}}{\deg A_{k}}}{\mathbb A}_{k-1}\,,
\end{align*}
whence one obtains  \eqref{eq:recursive-Jacobi}, since \eqref{eq:degree-Ak-distance} implies $a_k\deg A_k=b_k\deg A_{k-1}$.
\end{proof}

Consider the probability measure defined on the $\sigma$-algebra of Borel subsets of ${\mathbb R}$, given by 
\begin{gather}
\label{eq:measure-A-v}
\mu_A(B)\ceq\ip{v}{E_A(B)v}\,,
\end{gather}
where $E_A$ is the  spectral measure of $A$ and $v$ is a fixed vertex. 

\begin{remark}\label{rm:spectral distribution-a0}
For $k\in{\mathbb N}_0$, it is a simple matter to verify from the recursive relation \eqref{eq:recursive-Jacobi} that $A^k=\sum_{t=0}^k\beta_{k,t}{\mathbb A}_t$, with $\beta_{k,t}\in{\mathbb C}$. Thus, for any $\delta_i\in V$, 
\begin{gather}\label{eq:spectral distribution-a0}
\ip{\delta_i}{A^k\delta_i}=\sum_{t=0}^k\beta_{k,t}\ip{\delta_i}{{\mathbb A}_t\delta_i}=\beta_{k,0}\,.
\end{gather}
Then, one has by \eqref{eq:measure-A-v} and \eqref{eq:spectral distribution-a0} that 
\begin{gather*}
\ip{\delta_i}{A^k\delta_i}=\beta_{k,0}=\ip{v}{A^kv}=\int x^k d\mu_A\,,
\end{gather*}
viz. the spectral distribution of $A$ in a vertex, does not depend on $v\in V$. 
\end{remark}

In what follows, we will work in the Hilbert space $(\cA K, \ip{\cdot}{\cdot}_{\mu_A})$, where \begin{gather*}
\cA K=\lrb{f(A)\,:\, f\in L_2(\mathbb R,\mu_A)}\,,
\end{gather*}
which is isomorphic to $L_2(\mathbb R,\mu_A)$ (cf. \cite[sect.\,13.4]{MR1157815} and \cite[Sect.\,5.3]{MR2953553}). Thus, it follows because of Remark~\ref{rm:spectral distribution-a0} that
\begin{gather}
\label{eq:equality-inner}
\ip{f}{g}_{\mu_A}=\ip{f(A)v}{g(A)v}\,,\quad\mbox{for all}\quad f,g\in\cA K\,.
\end{gather}

\begin{remark}\label{rm:onb-K} The family $\{{\mathbb A}_k\}_{k\in{\mathbb N}_0}$ is an orthonormal basis for $(\cA K, \ip{\cdot}{\cdot}_{\mu_A})$. Indeed,
 \begin{align*}
 \ip{{\mathbb A}_i}{{\mathbb A}_j}_{\mu_A}=\frac{1}{\sqrt{\deg A_i\deg A_j}}\ip{A_i x}{A_j x}=\delta_{ij}\,,\quad i,j\in{\mathbb N}_0\,,
  \end{align*}
 where $\delta_{ij}$ is the Kronecker delta.  
\end{remark}

\begin{definition} The \emph{multiplication operator} $J$ in $\cA K$ is defined by 
\begin{align}\label{eq:multiplication-operator}
 \begin{split}
J\colon\textrm{dom}\, J&\to\cA K\\
f(A)&\mapsto Af(A)\,,
\end{split}
\end{align} where $\textrm{dom}\, J=\{f\in\cA K\,:\,f(A),Af(A)\in\cA K\}$.
\end{definition}

\begin{theorem}\label{th:Jacobi-Adjacency}
The multiplication operator $J$ is bounded and selfadjoint, with 
\begin{gather}
\label{ea:bound-J}
\no J_{\mu_A} \leq \deg A\,.
\end{gather}
Moreover, its matrix representation is a Jacobi matrix  given by 
\begin{gather}
\label{eq:matrix-represJ}
\begin{pmatrix}
0&\sqrt{a_1b_1}&0&0&\dots\\
\sqrt{a_1b_1}&\alpha_1&\sqrt{a_2b_2}&0&\dots\\
0&\sqrt{a_2b_2}&\alpha_2&\sqrt{a_3b_3}&\dots\\
0&0&\sqrt{a_3b_3}&\alpha_3&\dots\\
\dots&\dots&\dots&\dots&\dots
\end{pmatrix}\,,
\end{gather}
where $\{(a_n,b_n)\}_{n\in{\mathbb N}}$ is the  intersection of the graph and \begin{gather*}
\alpha_n=\deg A-(a_n+b_{n+1})\,,\quad n\in{\mathbb N}\,.
\end{gather*}

\end{theorem}
\begin{proof}
Since $A$ is selfadjoint, it follows that $J$ is symmetric. Moreover, inasmuch as $A$ is bounded in $\mathcal{H}$ and in view of \eqref{eq:equality-inner}, 
\begin{align}\label{eq:bound-of-JA}
\no{Af(A)}_{\mu_A}=\no{Af(A)v}\leq \no A\,\no {f(A)}_{\mu_A} <\infty\,,\quad  f\in\cA K\,,
\end{align}
which implies  $\textrm{dom}\, J=\cA K$. Besides, from \eqref{eq:norm-dist-regular} and \eqref{eq:bound-of-JA}, one yields 
\eqref{ea:bound-J}. So, we deduce that $J$ is selfadjoint and the family of complex polynomials  at $A$ is dense in $\cA K$ (cf. \cite[Sec.\,2]{MR0184042}). Moreover,  inasmuch as  $\{{\mathbb A}_k\}_{k\in{\mathbb N}_0}$ is an orthonormal basis for $\cA K$ (v.s. Remark~\ref{rm:onb-K}), the recursive relation \eqref{eq:recursive-Jacobi} implies \eqref{eq:matrix-represJ}. This completes the proof. 
 \end{proof}

\begin{remark}\label{Rm:spectral-properties} Since the family of complex polynomials at $A$ is dense in $\cA K$,  one gets 
\begin{align*}
\cc{\textrm{span}\, \lrb{ J^n {\mathbb A}_0}}_{n\in{\mathbb N}_0}=\cA K \,,
\end{align*}
viz.  ${\mathbb A}_0$ is a cycle vector and $J$ is simple (see \cite[Sec.\,69]{MR1255973}). 
Besides, the spectrum of $J$ is not purely discrete (cf. \cite[Prop.\,5.12]{MR2953553}) and is determined by  
 \begin{gather}\label{spectrum-J}
\sigma(J)=\textrm{supp}\, \mu_A\,.
 \end{gather}
Moreover, every eigenvalue $\lambda$ of $J$ is of multiplicity one, which coincides with $\mu_A(\{\lambda\})\neq0$  \cite[Sec.\, 4.7]{MR1192782} (also \cite[Sec.\,5.4]{MR2953553}). From Corollary~\ref{coro:bound-of-Ak} and \eqref{spectrum-J}, one has that $\no J_{\mu_A}=\deg A$ if and only the isoperimetric constant of the graph is equal to zero.  Since $\{(a_n,b_n)\}_{n\in{\mathbb N}}\subset{\mathbb N}^2$, it follows that 
$\lim_{n\rightarrow\infty}\sqrt{a_nb_n}\geq1$, whence one deduces that $J$ is not a compact operator (cf. \cite[Sec.\,28]{MR1255973}).
\end{remark}
\begin{corollary}\label{cor:Jacobi-bipartite}
The spectrum of $J$ is symmetric about zero if and only if
\begin{gather}\label{eq:no-alpha-Jacobi}
 b_{n+1}=\deg A-a_n\,,\quad\mbox{for all }n\in{\mathbb N}\,. 
\end{gather}
In such a case, $J$ has the following matrix representation
\begin{gather}\label{eq:matrix-represJ-sym}
\begin{pmatrix}
0&\sqrt{a_1\deg A}&0&\dots\\
\sqrt{a_1\deg A}&0&\sqrt{a_2(\deg A-a_1)}&\dots\\
0&\sqrt{a_2(\deg A-a_1)}&0&\dots\\
\dots&\dots&\dots&\dots
\end{pmatrix}\,.
\end{gather} 
\end{corollary}
\begin{proof}
We infer from  \eqref{spectrum-J}, items \eqref{eq:bip-ulf-2},\eqref{eq:bip-ulf-4} of Theorem~\ref{th:characterization-bigraph01} and the right-hand side of \eqref{eq:degree-Ak-distance} that  $\sigma(J)$ is symmetric about zero if and only $\isosc A_n=0$, for all $n\in{\mathbb N}$, which is true if and only if $0=\alpha_n=\deg A-(a_n+b_{n+1})$, i.e., \eqref{eq:no-alpha-Jacobi}. The representation \eqref{eq:matrix-represJ-sym} follows after replacing \eqref{eq:no-alpha-Jacobi} in \eqref{eq:matrix-represJ}.
\end{proof}

\section{Distance-regular graphs with finite diameter}
\label{sec:graphs-with-finite-diameter}
\noindent
We work in this section with a distance-regular graph with finite diameter $d\in{\mathbb N}$ (v.s. \eqref{eq:diameter-graph}). In this instance, its intersection sequence is $\{(a_k,b_k)\}_{k=1}^d$, since $A_k=0$, for all $k>d$.  Let $\cA K \ceq{\mathbb C}_d[A]$ denote the family of complex polynomials at $A$ of degree $\leq d$, endowed with the inner product $\ip{\cdot}{\cdot}_{\mu_A}$ given in \eqref{eq:equality-inner}. Thus, $\{{\mathbb A}_k\}_{k=0}^d$ is an orthonormal basis for $\cA K$, with ${\mathbb A}_k$ as in \eqref{eq:arthonormal-Ak} (see Remark~\ref{rm:onb-K}). 

In what follows, we shall tackle the problem of finding $\textrm{supp}\, \mu_A$ by means of extension theory for nondensely defined symmetric operators. So, we consider the symmetric operator $J$ with domain ${\cA K}\ominus{\{{\mathbb A}_d\}}$ and matrix representation
\begin{gather}\label{eq:J-non-densely-defined}
\begin{pmatrix}
0&\sqrt{a_1b_1}&0&\dots&0&0&*\\
\sqrt{a_1b_1}&\alpha_1&\sqrt{a_2b_2}&\dots&0&0&*\\
\dots&\dots&\dots&\dots&\dots&\dots\\
0&0&0&\dots&\alpha_{d-2}&\sqrt{a_{d-1}b_{d-1}}&*\\
0&0&0&\dots&\sqrt{a_{d-1}b_{d-1}}&\alpha_{d-1}&*\\
0&0&0&\dots&0&\sqrt{a_{d}b_{d}}&*
\end{pmatrix}\,,
\end{gather}
where $\alpha_n=\deg A-(a_n+b_{n+1})$, for $n=1,\dots,d-1$.  All the selfadjoint extensions of $J$ are characterized by 
\begin{gather}\label{eq:Jacobi-t}
J_\tau\ceq \begin{pmatrix}
0&\sqrt{a_1b_1}&0&\dots&0&0\\
\sqrt{a_1b_1}&\alpha_1&\sqrt{a_2b_2}&\dots&0&0\\
\dots&\dots&\dots&\dots&\dots&\dots\\
0&0&0&\dots&\alpha_{d-1}&\sqrt{a_{d}b_{d}}\\
0&0&0&\dots&\sqrt{a_{d}b_{d}}&\tau
\end{pmatrix}\,,\quad \tau\in{\mathbb R}\,
\end{gather}
which are adapted from  \cite[Thm.\,2.4]{MR1430397} (cf. \cite[Sec.\,5]{MR4091412}). 
\begin{remark}\label{rm:jinfnity-perturbation}
It is of interest to point out that in \cite{MR1430397,MR4091412} show another selfadjoint extension $J_\infty$ of $J$ which is not an operator. However, for a feasible analysis, we only work with the extensions \eqref{eq:Jacobi-t}, which satisfy 
\begin{gather*}
J_\tau f=J_0+\tau\ip{{\mathbb A}_d}{f}_{\mu_A}{\mathbb A}_d\,,\quad f\in\cA K
\end{gather*}
viz. $J_\tau$ is a one-rank perturbation of $J_0$. Consequently, for $j,k=0,\dots, d$, 
\begin{align}\label{eq:spectral-distributions-Jt}
\begin{split}
\int x^{j+k} d\mu_{J_\tau}&=\ip{I}{J_\tau^{j+k}I}_{\mu_A}=\ip{J_\tau^jI}{J_\tau^k I}_{\mu_A}\\&=\ip{J^jI}{J^k I}_{\mu_A}=\ip{A^j}{A^k }_{\mu_A}=\int x^{j+k} d\mu_A\,,
\end{split}
\end{align}
i.e., the spectral distributions of $J_\tau$ and $A$ coincide, for all $\tau\in{\mathbb R}$. 
\end{remark}
Now, let $\lambda\in{\mathbb R}$ and 
\begin{gather}\label{eq:eigenvalue-tau}
\varphi(A)=\sum_{k=0}^d\varphi_k{\mathbb A}_k \in\cA K\,,\quad (\varphi_k\in{\mathbb C})
\end{gather}
such that $J_\tau\varphi=\lambda\varphi$. Then, 
\begin{align}\label{eq:first-kind-polynomial0}
\begin{split}
\lambda\varphi_0-\sqrt{a_1b_1}\varphi_1&=0\,,\\
-\sqrt{a_{k}b_{k}}\varphi_{k}+(\lambda-\alpha_{k-1})\varphi_{k-1}-\sqrt{a_{k-1}b_{k-1}}\varphi_{k-2}&=0\,,\quad	(2\leq k\leq d)\\
(\lambda-\tau)\varphi_d-\sqrt{a_{d}b_{d}}\varphi_{d-1}&=0\,.
\end{split}
\end{align}
Clearly for $k=1,\dots,d$,  the number $\varphi_k$ is determined uniquely from  $\varphi_0$  and is a polynomial of degree $k-1$ at $\lambda$. Thereby,   \begin{gather}\label{eq:dim-eigenspace}
\dim \ker(J_\tau-\lambda I)\leq 1\,.
\end{gather}

We use the above reasoning to define the following. 
\begin{definition}
The  \emph{first-kind} polynomials associated to $J_\tau$ are defined by 
\begin{align}\label{eq:k-polynomial}
\begin{split}
P_0(x)&\ceq1\,,\\
P_1(x)&\ceq x/\sqrt{a_1b_1}\,,\\
P_{k}(x)&\ceq\frac{(x-\alpha_{k-1})P_{k-1}(x)-\sqrt{a_{k-1}b_{k-1}}P_{k-2}(x)}{\sqrt{a_{k}b_{k}}}\,,\quad (2\leq k\leq d)\\
P_{d+1}^{(\tau)}(x)&\ceq(x-\tau)P_d(x)-\sqrt{a_{d}b_{d}}P_{d-1}(x)\,.
\end{split}
\end{align}
 \end{definition}
The  polynomials \eqref{eq:k-polynomial} have real coefficients. Besides, $\{P_k\}_{k=0}^d$ is the same for any $J_\tau$, since $J_\tau$ is a one-rank perturbation of $J_0$.
\begin{theorem}\label{eq:to-spectra-Jtau}
For $\tau\in{\mathbb R}$, the spectrum of the selfadjoint extension $J_\tau$ is
\begin{gather*}
\sigma(J_\tau)=\lrb{\lambda^{(\tau)}\in{\mathbb R}\,:\,P_{d+1}^{(\tau)}(\lambda^{(\tau)})=0}\,.
\end{gather*}
Moreover, every eigenvalue $\lambda^{(\tau)}\in\sigma(J_\tau)$ is of multiplicity one and its corresponding eigenfunction (up to normalization) is  
\begin{gather}\label{eq:eigenfunctions-Jt}
\varphi_{\lambda^{(\tau)}}(A)=\sum_{k=0}^dP_k(\lambda^{(\tau)}){\mathbb A}_k\,.
\end{gather}
\end{theorem}

\begin{proof}
The  first part of proof is straightforward by remarking that $\{P_k(\lambda)\}_{k=0}^d$ holds \eqref{eq:first-kind-polynomial0} if and only if $P_{d+1}^{(\tau)}(\lambda)=0$. The multiplicity of every eigenvalue follows from \eqref{eq:dim-eigenspace}. The corresponding eigenvector  \eqref{eq:eigenfunctions-Jt} is directly from \eqref{eq:eigenvalue-tau}. 
\end{proof}

\begin{remark}
For $i=0,\dots,d$, the \emph{Christoffel-Darboux} kernel is  
\begin{gather}\label{eq:CD-Formulae}
K_i(x,y)\ceq\sum_{j=0}^{i}P_j(x)P_j(y)\,,
\end{gather}
which satisfies  (cf. \cite{MR587909})
\begin{align}\label{eq:reproducing-polynomials}
K_k(x,y)=\sqrt{a_{k+1}b_{k+1}}\frac{P_k(y)P_{k+1}(x)-P_k(x)P_{k+1}(y)}{x-y}\,,
\end{align}
for $k=0,\dots,d-1$. Moreover, it holds the following property  
\begin{align}\label{eq:kernel-sumpk}
K_d(x,x)={P_d(x)\lrp{P_{d+1}^{(\tau)}(x)}'-\lrp{P_d(x)}'P_{d+1}^{(\tau)}}(x)\,.
\end{align} 
Indeed, one simply computes from \eqref{eq:reproducing-polynomials} that 
\begin{align*}
K_d(x,y)&=P_d(x)P_d(y)+\sqrt{a_{d}b_{d}}\frac{P_{d-1}(y)P_{d}(x)-P_{d-1}(x)P_{d}(y)}{x-y}\\
&= P_d(y)\frac{P_{d+1}^{(\tau)}(x)}{x-y}-P_d(x)\frac{P_{d+1}^{(\tau)}(y)}{x-y}\,,
\end{align*}
wherefrom letting $x$ tends to $y$, one yields \eqref{eq:kernel-sumpk}.
\end{remark}

\begin{corollary}
The spectra of the selfadjoint extensions $J_\tau$ have no intersection and are pairwise interlaced.
\end{corollary}
\begin{proof}
For $\tau\in{\mathbb R}$, one has from Theorem~\ref{eq:to-spectra-Jtau} that the eigenvalues of $J_\tau$ are the roots of $P_{d+1}^{(\tau)}$, which are real and different from each other. So for $\eta\neq\tau$, if $\lambda$ is a root of  $P_{d+1}^{(\tau)}$, then 
\begin{align}\label{eq:Pn1vsPn}
P_{d+1}^{(\eta)}(\lambda)&=P_{d+1}^{(\eta)}(\lambda)-P_{d+1}^{(\tau)}(\lambda)=(\tau-\eta)P_{d}(\lambda)\,.
\end{align}
Besides, in view of \eqref{eq:kernel-sumpk}, 
\begin{align}\label{eq:kernel-prime}
P_d(\lambda)(P_{d+1}^{(\tau)})'(\lambda)=K_d(\lambda,\lambda)>0\,,
\end{align}
which by \eqref{eq:Pn1vsPn} $P_{d+1}^{(\eta)}(\lambda)\neq0$, viz. $J_\tau$ and $J_\eta$ have no common eigenvalues. 

Now, if $\alpha<\beta$ are two consecutive eigenvalues of $J_\tau$, then one has that $\textrm{sgn}\,(P_{n+1}^{(\tau)})'(\alpha)\neq \textrm{sgn}\, (P_{n+1}^{(\tau)})'(\beta)$ and \eqref{eq:kernel-prime} yields $\textrm{sgn}\, P_{d}(\alpha)\neq \textrm{sgn}\, P_{d}(\beta)$. Thus, \eqref{eq:Pn1vsPn} implies $\textrm{sgn}\, P_{d+1}^{(\eta)}(\alpha)\neq \textrm{sgn}\, P_{d+1}^{(\eta)}(\beta)$, which provide that  $J_\eta$ has an eigenvalue in $(\alpha,\beta)$. To conclude, if $J_\eta$ has two consecutive eigenvalues $\gamma_1<\gamma_2$ within $(\alpha,\beta)$, then one infers using the same above reasoning that  $J_\tau$ has an eigenvalue in $(\gamma_1,\gamma_2)$. This contradicts our assumption that $\alpha,\beta$ are consecutive.
\end{proof}

The above reasoning shows that there is a one-to-one correspondence, except at one point without considering the selfadjoint extension $J_\infty$ of \eqref{eq:J-non-densely-defined} (q.v. Remark~\ref{rm:jinfnity-perturbation}),  between the interval of two consecutive eigenvalues $\alpha<\beta$ of $J_{\tau_0}$ and the set  $\{\lambda_{J_\tau}\in\sigma(J_\tau)\cap(\alpha,\beta)\}_{\tau_0\neq\tau\in{\mathbb R}}$. Roughly speaking, Fig.~\ref{fig:eigenJtau} represents this behavior.  
\begin{figure}[h]\centering
\includegraphics[width=8cm]{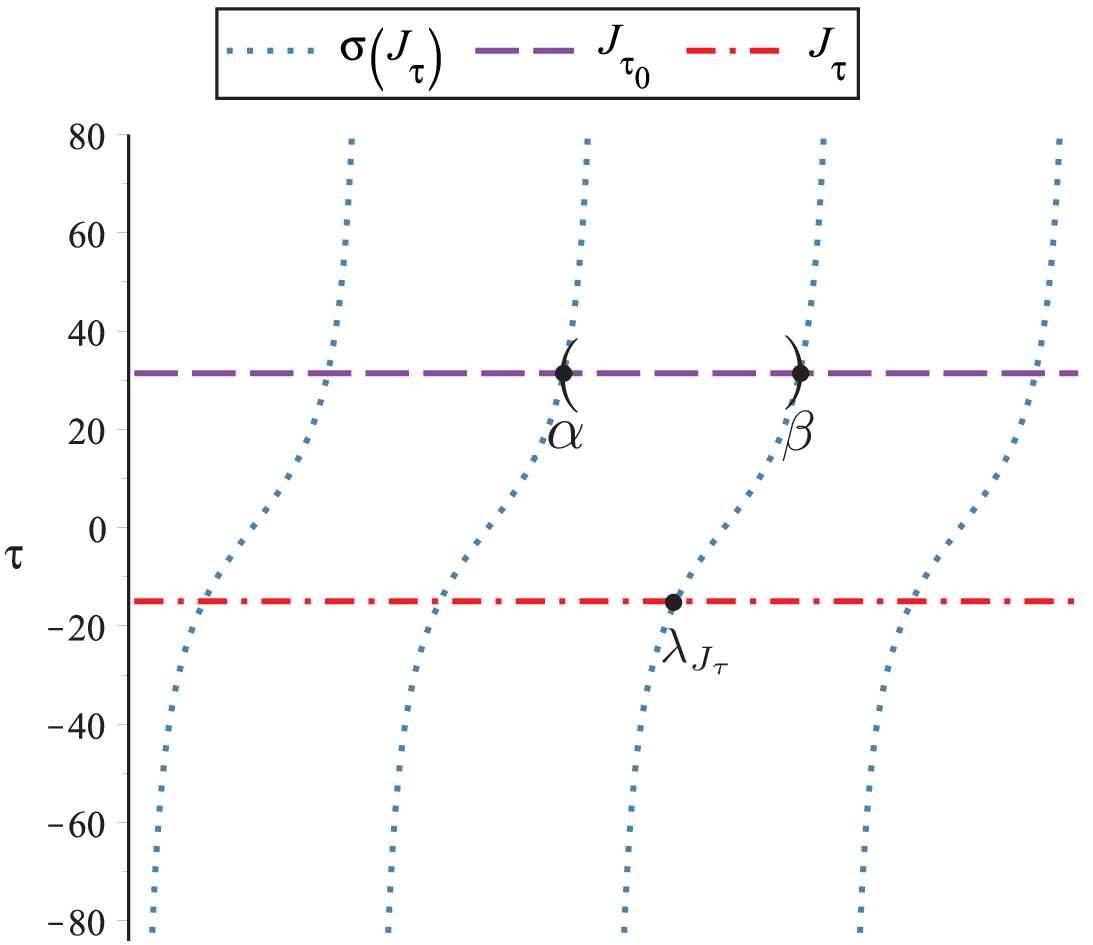}
  \caption{\centering Eigenvalues of $J_\tau$.}
  \label{fig:eigenJtau}
\end{figure}

For $\tau\in{\mathbb R}$, one has on the basis of \eqref{eq:spectral-distributions-Jt} that $\mu_{J_\tau}$ is a probability measure. Besides, it follows from Theorem~\ref{eq:to-spectra-Jtau} that  
\begin{align}\label{eq:spectral-measure-jt}
\mu_{J_\tau}(x)=\sum_{\lambda\in\sigma(J_\tau)}\frac1{
{\no{\varphi_{\lambda}(A)}}_{\mu_A}^{2}}{\mathds{1}_{\lambda}(x)}=\sum_{\lambda\in\sigma(J_\tau)}\frac1{K_d(\lambda,\lambda)}{\mathds{1}_{\lambda}(x)}\,.
\end{align}

In the following, we will determine the support of the measure of $A$. 

\begin{lemma}\label{lem:minimal-A-normalized}
For $k=0,\dots,d$, the $k$-th first-kind polynomial satisfies 
\begin{gather*}
P_k(A)={\mathbb A}_k\,.
\end{gather*}
\end{lemma}
\begin{proof}
The proof carries out by induction over $k$. It is clear that $P_0(A)={\mathbb A}_0$ and $P_1(A)={\mathbb A}_1$, since $a_1=1$ and $b_1=\deg A$. So, we may suppose that $P_{j}(A)={\mathbb A}_{j}$, for $j=0,\dots,k-1$. Thus by \eqref{eq:k-polynomial}, 
\begin{align*}
P_{k}(A)&=\frac{(A-\alpha_{k-1} I)P_{k-1}(A)-\sqrt{a_{{k-1}}b_{{k-1}}}P_{k-2}(A)}{\sqrt{a_{k}b_{k}}}\\
&=\frac{1}{\sqrt{a_{k}b_{k}}}\left(A{\mathbb A}_{k-1}-\alpha_{k-1}{\mathbb A}_{k-1}-\sqrt{a_{k-1}b_{k-1}}{\mathbb A}_{k-2}\right)\,,
\end{align*}
whence from Proposition~\ref{prop:recursive-normalized}, the assertion follows.
\end{proof}

The spectrum $\sigma(A)$ has $d+1$ distinct eigenvalues, since the diameter of the distance-regular graph is $d$  \cite[Sec.\,6.3]{MR3617953}. So, the degree of the minimal polynomial of $A$ is $d+1$.
\begin{theorem}\label{Thm:Jacobi-vs-A}
The support of $\mu_A$ is the spectrum of the extension  $J_{\deg A-a_d}$. \end{theorem}
\begin{proof}
We only need to show that $A$ and $J_{\deg A-a_d}$ have the same minimal polynomial. Note from \eqref{eq:cond-dist-regular} that $b_k=0$, since ${\mathbb A}_k=0$, for all $k> d$. Thus,  by virtue of Proposition~\ref{prop:recursive-normalized}, 
 \begin{gather}\label{eq:finite-n-recurrence}
A{\mathbb A}_d=\left(\deg A-a_d\right){\mathbb A}_d+\sqrt{a_db_d}{\mathbb A}_{d-1}\,.
\end{gather}
Moreover, Theorem~\ref{eq:to-spectra-Jtau} asserts that $P_{d+1}^{(\deg A-a_d)}$ is the minimal polynomial of $J_{\deg A-a_d}$. In this fashion, from Lemma~\ref{lem:minimal-A-normalized} and in view of \eqref{eq:finite-n-recurrence}, one computes 
\begin{align*}
P_{d+1}^{(\deg A-a_d)}(A)&=(A-(\deg A-a_d)I)P_d(A)-\sqrt{a_{d}b_{d}}P_{d-1}(A)\\
&=A{\mathbb A}_d-(\deg A-a_d){\mathbb A}_d-\sqrt{a_{d}b_{d}}{\mathbb A}_{d-1}=0\,,
\end{align*}
which completes the proof.
\end{proof}

We conclude this section with the following result, which is known as \emph{Biggs' formula}, and it was first shown in \cite[Th.\,21.4]{MR1271140}.
\begin{corollary}\label{cor:multiplicity-eigenvalues-A}
Let $\{\lambda_i\}_{i=0}^d$ be the distinct eigenvalues of $A$, with respectively multiplicities $\{m(\lambda_i)\}_{i=0}^d$ and $n\in{\mathbb N}$ the number of vertices of the graph. Then, 
\begin{gather*}
m(\lambda_i)=\frac{n}{K_d(\lambda_i,\lambda_i)}\,, \quad 0\leq i\leq d
\end{gather*}
where $K_d$ is the Christoffel-Darboux kernel \eqref{eq:CD-Formulae}.
\end{corollary}
\begin{proof}
Note that $A$ is an $n\times n$ matrix and $\mu_A(x)=n^{-1}\sum_{i=0}^dm(\lambda_i)\mathds{1}_{\lambda_i}(x)$, since it is a probability measure. Hence, it follows from \eqref{eq:spectral-measure-jt} and Theorem~\ref{Thm:Jacobi-vs-A} that 
\begin{gather*}
\frac{m(\lambda_i)}{n}=\mu_A(\lambda_i)=\mu_{J_{\deg A-a_d}}(\lambda_i)=\frac{1}{K_d(\lambda_i,\lambda_i)}\,,\quad (0\leq i\leq d)
\end{gather*}
as required.
\end{proof}

\section{Examples}
\label{sec:Examples}
\subsection{Regular trees}
For a fix number $n\geq2$,  let $T_n$ denote a tree in which each vertex has exactly $n$ neighbors, viz. its adjacency operator $A$ is $n$-regular (e.g., see Fig.~\ref{fig:Tn-regular-tree}).
\begin{figure}[h!]\centering
    \includegraphics[width=12cm]{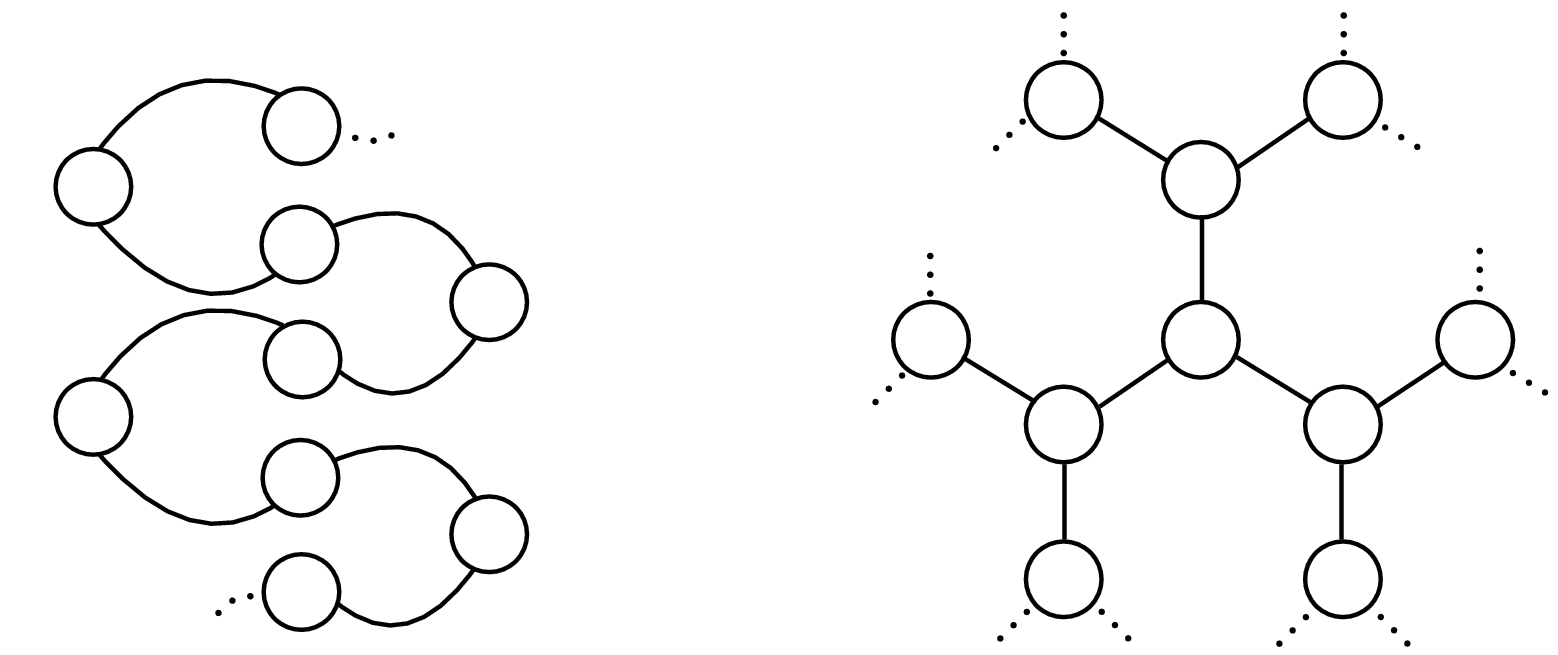}
    \caption{\centering $T_2$ and $T_3$.}
  \label{fig:Tn-regular-tree}
\end{figure}
The graph $T_n$  is  distance-regular with intersection sequence 
\begin{gather*}
\lrb{(1,n),(1,n-1),(1,n-1),\dots}\,.
\end{gather*}
Thus, from Theorem~\ref{th:recursive-Ak}, every $k$-adjacency operator is regular and isoscyclical, with 
\begin{gather*}
\deg A_k=n(n-1)^{k-1}\quad\mbox{and}\quad \isosc A_k=0\,.\quad (k\geq 1)
\end{gather*}
Besides, the Jacobi adjacency operator \eqref{eq:matrix-represJ} of  $T_n$ is 
\begin{gather*}
J_{T_n}=\begin{pmatrix}
0&\sqrt{n}&0&0&\dots\\
\sqrt{n}&0&\sqrt{n-1}&0&\dots\\
0&\sqrt{n-1}&0&\sqrt{n-1}&\dots\\
0&0&\sqrt{n-1}&0&\dots\\
\dots&\dots&\dots&\dots&\dots
\end{pmatrix}\,,
\end{gather*}
with spectral distribution (equivalent to the spectral distribution of $A$) 
\begin{gather}\label{eq:spectral-distribution-A}
d\mu_A(x)=\frac{n\sqrt{4(n-1)-x^2}}{2\pi(n^2-x^2)}dx\,,\quad \abs x\leq2\sqrt{n-1}\,.
\end{gather}
Clearly, \eqref{eq:spectral-distribution-A} is symmetric about zero, i.e., $T_n$ is bipartite. Moreover, the norm of $\no{J_{T_n}}_{\mu_A}$ is $2\sqrt{n-1}$ (v.s. Remark~\ref{Rm:spectral-properties} and Corollary~\ref{cor:Jacobi-bipartite}).  

The spectral distribution \eqref{eq:spectral-distribution-A} was proven by Mckay \cite{MR629617} (cf. \cite[Sec.\,6.5]{MR3617953}). Besides, in \cite{MR3542847} shows the distribution of all the $k$-adjacency operators of $T_n$.

\subsection{Complete graphs}
For a solid analysis, let us regard a fixed number $n\geq2$ and denote by $K_n$ the graph with $n$ vertices in which any pair of vertices are adjacent to each other (e.g., see Fig.~\ref{fig:Kn-regular}).
\begin{figure}[h!]\centering
    \includegraphics[width=14cm]{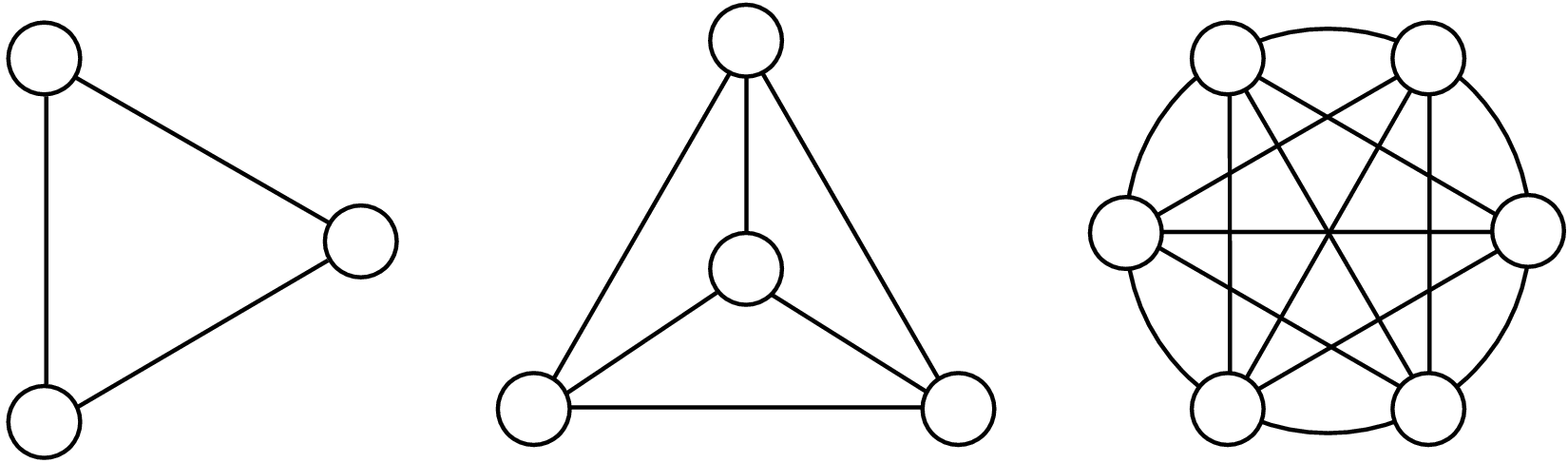}
  \caption{\centering $K_3$, $K_4$ and $K_6$.}
  \label{fig:Kn-regular}
\end{figure}
The graph $K_n$ is distance-regular with diameter equal to one and intersection sequence $\{(1,n-1)\}$. In view of Theorem~\ref{th:recursive-Ak}, the adjacency operator $A$ of $K_n$ follows 
 \begin{gather*}
 \deg A=n-1\quad \mbox{and}\quad \isosc A=\frac{(n-2)(n-1)}{2}\,,
 \end{gather*}
 e.g., any vertex of $K_6$ has five neighbors and ten isoscycles.
 
The operator \eqref{eq:J-non-densely-defined} and its selfadjoint extensions \eqref{eq:Jacobi-t},  in the Hilbert  space $\cA K=\textrm{span}\,\{ I,{\mathbb A}_1\}$, are given by
\begin{gather*}
J=\begin{pmatrix}
0&*\\
\sqrt{n-1}&*
\end{pmatrix}\quad;\quad
J_\tau=\begin{pmatrix}
0&\sqrt{n-1}\\
\sqrt{n-1}&\tau
\end{pmatrix}\,,\quad (\tau\in{\mathbb R})
\end{gather*} 
respectively. The eigenvalues of $\sigma(J_\tau)=\{\lambda_+^{(\tau)},\lambda_-^{(\tau)}\}$ are characterized by 
\begin{gather}\label{eq:eigenvalue-function}
\lambda_+^{(\tau)}=\frac{\tau}{2}+\frac{1}{2}\sqrt{\tau^2+4(n-1)}\quad;\quad \lambda_-^{(\tau)}=\frac{\tau}{2}-\frac{1}{2}\sqrt{\tau^2+4(n-1)}\,,
\end{gather}
with eigenfunctions (up to normalization) $\varphi_{\lambda_\pm^{(\tau)}}(A)=I+(n-1)^{-1}\lambda_\pm^{(\tau)}A$. Besides, the spectral measures \eqref{eq:spectral-measure-jt} are
\begin{gather}\label{eq:spectral-measures-Jt}
\mu_{J_\tau}(x)=\frac{1}{1+\lambda_+}\mathds{1}_{\lambda_+}(x)+\frac{1}{1+\lambda_-}\mathds{1}_{\lambda_-}(x)\,,
\end{gather}
which have the same distribution  equal to the distribution of  $\mu_A$ (v.s. \eqref{eq:spectral-distributions-Jt}).

Now, on the basis of Theorem~\ref{Thm:Jacobi-vs-A}, the support of $\mu_A$ is (e.g., see Fig.~\ref{fig:Eigenvalues-Kn})
\begin{figure}[h!]\centering
    \includegraphics[width=16cm]{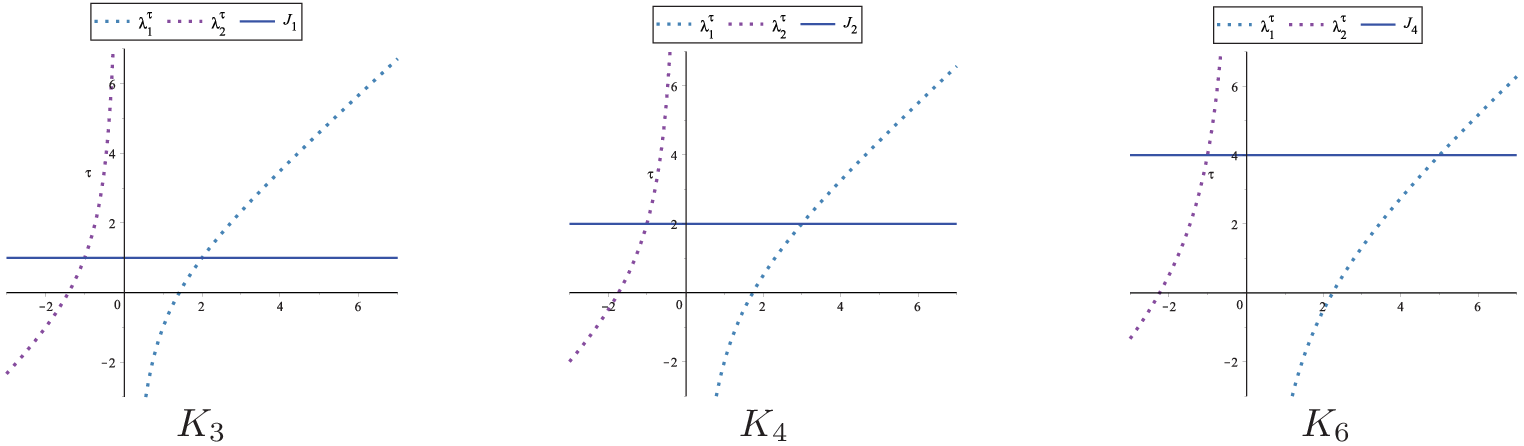} 
  \caption{\centering Eigenvalues \eqref{eq:eigenvalue-function} vs. $J_{n-2}$.}
  \label{fig:Eigenvalues-Kn}
\end{figure}
\begin{gather*}
\sigma(J_{n-2})=\lrb{-1,n-1}\,.
\end{gather*}
In this fashion, \eqref{eq:eigenvalue-function} and \eqref{eq:spectral-measures-Jt} yield
\begin{align*}
\mu_A(x)=\frac{n-1}{n}\mathds{1}_{-1}(x)+\frac{1}{n}\mathds{1}_{n-1}(x)\,.
\end{align*}

The above account clarifies that the eigenvalues of $A$ are $ -1$ and $n-1$ with multiplicities $ n-1$ and $1$, respectively (v.s. Corollary~\ref{cor:multiplicity-eigenvalues-A}).\vskip3mm

\noindent{\bf Acknowledgments:}
\noindent
\noindent The author gratefully acknowledges the partial supports from CONACYT-Mexico, FORDECYT 265667 “Programa para un Avance Global e Integrado de la Matem\'atica Mexicana” and postdoctoral fellowship 136135 “Estructura de los estados estacionarios de generadores de Markov de transporte cu\'antico y del l\'imite de baja densidad”. The author thanks Prof. O. Arizmendi for useful comments of the exposition of this material. 



\def\cprime{$'$} \def\lfhook#1{\setbox0=\hbox{#1}{\ooalign{\hidewidth
  \lower1.5ex\hbox{'}\hidewidth\crcr\unhbox0}}} \def\cprime{$'$}
  \def\cprime{$'$} \def\cprime{$'$} \def\cprime{$'$} \def\cprime{$'$}
  \def\cprime{$'$} \def\cprime{$'$}
\providecommand{\bysame}{\leavevmode\hbox to3em{\hrulefill}\thinspace}
\providecommand{\MR}{\relax\ifhmode\unskip\space\fi MR }
\providecommand{\MRhref}[2]{%
  \href{http://www.ams.org/mathscinet-getitem?mr=#1}{#2}
}
\providecommand{\href}[2]{#2}


\begin{thebibliography}{10}

\bibitem{MR0184042}
N.~I. Akhiezer, \emph{The classical moment problem and some related questions
  in analysis}, Translated by N. Kemmer, Hafner Publishing Co., New York, 1965.
  \MR{0184042 (32 \#1518)}

\bibitem{MR1255973}
N.~I. Akhiezer and I.~M. Glazman, \emph{Theory of linear operators in {H}ilbert
  space}, Dover Publications Inc., New York, 1993, Translated from the Russian
  and with a preface by Merlynd Nestell, Reprint of the 1961 and 1963
  translations, Two volumes bound as one. \MR{1255973 (94i:47001)}

\bibitem{MR0123188}
Richard Arens, \emph{Operational calculus of linear relations}, Pacific J.
  Math. \textbf{11} (1961), 9--23. \MR{0123188 (23 \#A517)}

\bibitem{MR3542847}
Octavio Arizmendi and Tulio Gaxiola, \emph{On the spectral distribution of
  distance-{$k$} graph of free product graphs}, Infin. Dimens. Anal. Quantum
  Probab. Relat. Top. \textbf{19} (2016), no.~3, 1650017, 17. \MR{3542847}

\bibitem{MR924236}
N.~L. Biggs, Bojan Mohar, and John Shawe-Taylor, \emph{The spectral radius of
  infinite graphs}, Bull. London Math. Soc. \textbf{20} (1988), no.~2,
  116--120. \MR{924236}

\bibitem{MR1271140}
Norman Biggs, \emph{Algebraic graph theory}, second ed., Cambridge Mathematical
  Library, Cambridge University Press, Cambridge, 1993. \MR{1271140}

\bibitem{MR1192782}
M.~Sh. Birman and M.~Z. Solomjak, \emph{Spectral theory of selfadjoint
  operators in {H}ilbert space}, Mathematics and its Applications (Soviet
  Series), D. Reidel Publishing Co., Dordrecht, 1987, Translated from the 1980
  Russian original by S. Khrushch{\"e}v and V. Peller. \MR{1192782 (93g:47001)}

\bibitem{MR1631548}
Ronald Cross, \emph{Multivalued linear operators}, Monographs and Textbooks in
  Pure and Applied Mathematics, vol. 213, Marcel Dekker, Inc., New York, 1998.
  \MR{1631548}

\bibitem{MR1318517}
V.~A. Derkach and M.~M. Malamud, \emph{The extension theory of {H}ermitian
  operators and the moment problem}, J. Math. Sci. \textbf{73} (1995), no.~2,
  141--242, Analysis. 3. \MR{1318517}

\bibitem{MR0361889}
A.~Dijksma and H.~S.~V. de~Snoo, \emph{Self-adjoint extensions of symmetric
  subspaces}, Pacific J. Math. \textbf{54} (1974), 71--100. \MR{0361889 (50
  \#14331)}

\bibitem{MR3869419}
Pavel Exner, Aleksey Kostenko, Mark Malamud, and Hagen Neidhardt,
  \emph{Spectral theory of infinite quantum graphs}, Ann. Henri Poincar\'{e}
  \textbf{19} (2018), no.~11, 3457--3510. \MR{3869419}

\bibitem{MR1430397}
Seppo Hassi and Henk de~Snoo, \emph{One-dimensional graph perturbations of
  selfadjoint relations}, Ann. Acad. Sci. Fenn. Math. \textbf{22} (1997),
  no.~1, 123--164. \MR{1430397 (97m:47025)}

\bibitem{MR3902704}
Feng Jian and Dandan. Shi, \emph{Complex network theory and its application
  research on {P}2{P} networks}, Appl. Math. Nonlinear Sci. \textbf{1} (2016),
  no.~1, 45--52. \MR{3902704}

\bibitem{MR3891807}
Aleksey Kostenko and Noema Nicolussi, \emph{Spectral estimates for infinite
  quantum graphs}, Calc. Var. Partial Differential Equations \textbf{58}
  (2019), no.~1, Art. 15, 40. \MR{3891807}

\bibitem{MR587909}
H.~J. Landau, \emph{The classical moment problem: {H}ilbertian proofs}, J.
  Functional Analysis \textbf{38} (1980), no.~2, 255--272. \MR{587909}

\bibitem{MR629617}
Brendan~D. McKay, \emph{The expected eigenvalue distribution of a large regular
  graph}, Linear Algebra Appl. \textbf{40} (1981), 203--216. \MR{629617}

\bibitem{MR683222}
Bojan Mohar, \emph{The spectrum of an infinite graph}, Linear Algebra Appl.
  \textbf{48} (1982), 245--256. \MR{683222}

\bibitem{MR3617953}
Nobuaki Obata, \emph{Spectral analysis of growing graphs}, SpringerBriefs in
  Mathematical Physics, vol.~20, Springer, Singapore, 2017, A quantum
  probability point of view. \MR{3617953}

\bibitem{MR4091412}
Josu\'{e}~I. Rios-Cangas and Luis~O. Silva, \emph{Perturbation theory for
  selfadjoint relations}, Ann. Funct. Anal. \textbf{11} (2020), no.~1,
  154--170. \MR{4091412}

\bibitem{MR1157815}
Walter Rudin, \emph{Functional analysis}, second ed., International Series in
  Pure and Applied Mathematics, McGraw-Hill, Inc., New York, 1991. \MR{1157815
  (92k:46001)}

\bibitem{MR2953553}
Konrad Schm{\"u}dgen, \emph{Unbounded self-adjoint operators on {H}ilbert
  space}, Graduate Texts in Mathematics, vol. 265, Springer, Dordrecht, 2012.
  \MR{2953553}

\bibitem{MR3823031}
Allen Tannenbaum, \emph{Robustness of complex networks with applications to
  cancer biology}, SIAM News \textbf{51} (2018), no.~3, 1, 3. \MR{3823031}

\bibitem{MR1711536}
Gerald Teschl, \emph{Jacobi operators and completely integrable nonlinear
  lattices}, Mathematical Surveys and Monographs, vol.~72, American
  Mathematical Society, Providence, RI, 2000. \MR{1711536}

\bibitem{MR657116}
Aleksandar Torga\v{s}ev, \emph{On spectra of infinite graphs}, Publ. Inst.
  Math. (Beograd) (N.S.) \textbf{29(43)} (1981), 269--282. \MR{657116}

\end{thebibliography}


\end{document}